\documentclass[conference]{IEEEtran}
\IEEEoverridecommandlockouts
\usepackage{cite}
\usepackage{amsmath,amssymb,amsfonts}
\usepackage{graphicx}
\usepackage{textcomp}
\usepackage[utf8]{inputenc}
\def\BibTeX{{\rm B\kern-.05em{\sc i\kern-.025em b}\kern-.08em
    T\kern-.1667em\lower.7ex\hbox{E}\kern-.125emX}}

\usepackage{amsmath,amssymb,amsthm}
\usepackage{mathtools}
\usepackage{nicefrac}
\usepackage{stmaryrd}
\usepackage{thmtools}

\usepackage{enumitem}
\usepackage[normalem]{ulem}
\usepackage{graphicx}

\theoremstyle{theorem}
\newtheorem{theorem}{Theorem}
\theoremstyle{definition}

\theoremstyle{example}
\newtheorem{example}{Example}
\theoremstyle{lemma}
\newtheorem{lemma}{Lemma}
\theoremstyle{corollary}

\theoremstyle{claim}

\theoremstyle{remark}

\theoremstyle{proposition}
\newtheorem{proposition}{Proposition}
\theoremstyle{algorithm}

\theoremstyle{algorithm*}
\newtheorem*{algorithm*}{Algorithm}

\newlist{thmlist}{enumerate}{1}
\setlist[thmlist]{label=(\alph{thmlisti}), ref=\thetheorem(\alph{thmlisti}),noitemsep}

\usepackage{xspace}
\usepackage{complexity}
\usepackage{subcaption}
\usepackage[labelsep=period]{caption}
\usepackage{tikz}
\usetikzlibrary{shapes}
\usetikzlibrary{decorations.pathmorphing}
\usetikzlibrary{decorations.pathreplacing}

\tikzset{
	bluematch/.style   = {double, double distance=1.5pt, thick, >=stealth },				
	redheavymatch/.style   = {line width=3pt },
	redmatch/.style   = {line width=1pt, >=stealth },
	tinycirc/.style = {draw, circle, inner sep=0pt, minimum width=6pt},
	light/.style   = {densely dotted, line width=1pt, >=stealth }
}

\usepackage[color=white, linecolor=lightgray,prependcaption]{todonotes} %

\usepackage[linkbordercolor=red,citebordercolor=blue,urlbordercolor=green,pagebackref]{hyperref}
\usepackage[nameinlink,noabbrev]{cleveref}
\crefname{subsection}{subsection}{subsections}
\crefname{Subsection}{Subsection}{Subsections}
\crefname{figure}{Fig.}{Figs.}

\usepackage{lineno}
\newcommand{\defeq}{\vcentcolon=}

\newcommand{\wmax}{w_{\max}}

\newcommand{\T}[2]{T^{(#2)}_{#1}}

\newcommand{\eg}{e.g.,\xspace}
\newcommand{\ie}{i.e.,\xspace}

\newcommand{\Eopt}{E_{\text{opt}}}
\newcommand{\Esub}{E_{\text{sub}}}

\makeatletter%
\begin{document}

\title{Max-Product for Maximum Weight Matching -- Revisited}

\author{\IEEEauthorblockN{Mario Holldack}
\IEEEauthorblockA{\textit{Institut für Informatik} \\
\textit{Goethe-Universität}\\
Frankfurt am Main, Germany\\
holldack@thi.cs.uni-frankfurt.de}
}

\maketitle

\begin{abstract}
	We focus on belief propagation for the assignment problem, also known as the maximum weight bipartite matching problem. 
	We provide a constructive proof that the well-known upper bound on the number of iterations (Bayati, Shah, Sharma 2008) is tight up to a factor of four. 
	Furthermore, we investigate the behavior of belief propagation when convergence is not required.
	We show that the number of iterations required for a sharp approximation consumes a large portion of the convergence time. 
	Finally, we propose an ``approximate belief propagation'' algorithm for the assignment problem.
\end{abstract}


\section{Introduction}
\label{sec:introduction}

Since Pearl's introduction of the \emph{belief propagation algorithm} (BP) in \cite{pearl1982reverend}, applications of BP have been extensively covered in the literature, ranging from artificial intelligence, computer vision, communication, and combinatorial optimization to statistical physics; see \cite{DBLP:journals/tit/KschischangFL01} for an introductory survey. 
The same algorithm is also known as the \emph{max-product}, \emph{max-sum}, or \emph{sum-product algorithm} among others
Here we address the application of BP -- that is the max-sum algorithm -- to the assignment problem in a weighted complete bipartite graph $K_{n,n}$, \ie the problem of assigning $n$ jobs to $n$ employees such that every job is assigned exactly once and the profit is maximized. The assignment problem is also known as the maximum weight matching problem in a weighted complete bipartite graph. 
Here it is sufficient to know that BP is an iterative graph algorithm where each node outputs a local solution (a so-called \emph{belief}) in every iteration. 
\textcolor{black}{More precisely, a local solution of a node $u$ is an edge $\{u,v\}$ that $u$ believes to be in a \emph{maximum weight matching} (MWM). }
The algorithm stops when all local solutions converge, that is when the outputs no longer change. 
In~\cite{DBLP:journals/tit/BayatiSS08}  Bayati, Shah, and Sharma show that BP converges to the MWM within ${2n{\cdot}\wmax}/{\varepsilon}$ iterations, where $\wmax \defeq \max \{ |w_e| : e \in E \}$ and $\varepsilon$ is the \emph{uniqueness gap}, \ie the difference between the sum of the weights of the best and the second best perfect matching. 
In total their BP implementation takes $\mathcal{O}({n^3{\cdot}\wmax}/{\varepsilon})$ operations for finding the unique MWM which is 
comparable with the best known sequential algorithms -- given that $\wmax$ and $\varepsilon$ are fixed parameters. 
As shown by Salez and Shah in~\cite{DBLP:journals/mor/SalezS09}, BP is an optimal algorithm for the MWM problem in complete bipartite graphs with randomly weighted edges, \ie with high probability BP finds the maximum weight matching within a constant number of iterations. 

In \Cref{thm:one-cycle} we show that the upper bound \cite{DBLP:journals/tit/BayatiSS08} of ${2n {\cdot}\wmax}/{\varepsilon}$ iterations for the \emph{convergence time} is is tight up to a factor of four.
Based on this result we construct weights for the $K_{n,n}$ such that BP does not find any good approximate MWM, even when the number of iterations is close to the convergence time. 
What is the reason behind this surprisingly poor approximation behavior? One possible explanation is that the \emph{BP matching}, \ie the set of edges for which the beliefs of the endpoints agree, consists only of few edges. We show in \Cref{thm:union-cycles} that any completion of a BP matching computed in an early iteration has a poor approximation factor.

The rest of this paper is organized as follows: \Cref{sec:bp-for-the-assignment-problem} describes our main results. \Cref{sec:one-cycle} and \Cref{sec:union-cycles} cover the proofs of \Cref{thm:one-cycle} and \Cref{thm:union-cycles}, respectively. \Cref{sec:approximate-bp} presents an approximate BP algorithm and \Cref{sec:conclusions} concludes the paper.

\section{BP for the Assignment Problem}
\label{sec:bp-for-the-assignment-problem}
Let $K_{n,n}$ be the complete bipartite graph with $n$ nodes in each layer.
In \cite{DBLP:journals/tit/BayatiSS08} Bayati, Shah, and Sharma implement and analyze BP for the assignment problem on $K_{n,n}$ where edges receive real-valued weights. Their result is one of most important success stories of \emph{Loopy BP}, \ie BP on graphs with cycles. 
In the following $\wmax$ is the maximum absolute value of any edge weight and $\varepsilon$ is the difference (\emph{uniqueness gap}) between the sum of the weights of the best and the second best perfect matching. 

\begin{theorem}[Bayati, Shah, Sharma, \cite{DBLP:journals/tit/BayatiSS08}]
	\label{thm:bayatishahsharma}
	For any edge weights for the $K_{n,n}$, the BP algorithm converges to the maximum weight matching within $\tfrac{2n {\cdot} \wmax}{ \varepsilon}$ iterations, provided  the maximum weight matching is unique.
\end{theorem}
How tight is their analysis?
\begin{theorem}\label{thm:one-cycle}
	For any $n \geq 3$, $\wmax > 0$, and $0<\varepsilon < \frac{\wmax}{4(n–2)} $ there are edge weights for the $K_{n,n}$ such that the maximum weight matching is unique and BP converges to the maximum weight matching only after $\tfrac{n {\cdot} \wmax}{2\varepsilon}$ iterations.
\end{theorem}
\goodbreak
Thus, the bound of \Cref{thm:bayatishahsharma} cannot be improved. 
Since the ratio $\tfrac{\wmax}{\varepsilon}$ can be exponentially large in the number of input bits, \Cref{thm:one-cycle} implies that BP has an exponential worst-case convergence time.
However, demanding convergence may be too harsh since the algorithm may have found an approximate MWM (or even the MWM itself) already after relatively few iterations. 
Observe that in each iteration BP produces a partial matching consisting of all edges $\{\alpha, \beta\}$ where both endpoints believe that $\{\alpha, \beta\}$ belongs to the MWM. Hence, it is important to determine whether those partial matchings already constitute good approximations of the MWM. 
Maybe such a partial matching is not good enough, but can be completed into a good perfect matching with little additional resources. 
However, in the worst case, a sharp approximation cannot be achieved much earlier than convergence.

\begin{theorem}\label{thm:union-cycles}
	For sufficiently large $n$, for all $\wmax>0$ and $0<\varepsilon<\tfrac{\wmax}{4(n-2)}$, there are edge weights for the $K_{n,n}$ such that  every completion of a partial BP matching computed during the first 
	\[
		\min \left\{ \big(n\log(n)\big)^{\Theta\big(\sqrt{n/\log(n)}\big)} ~,~
		\Theta\Big( \tfrac{\wmax}{\sqrt{n^3/\log(n)}\cdot \varepsilon} \Big) \right\}
	\]
	iterations is $\big(1–\Theta(\nicefrac{1}{\sqrt{n\log(n)}})\big)$-approximative.
\end{theorem}

We construct weights such that the partial BP matchings are almost perfect, but none of the few completions are capable of improving the matching considerably. Observe that the sharp lower bound $\tfrac{n {\cdot} \wmax}{2\varepsilon}$ for the convergence time and the time $\Theta\big( {\wmax}/{\big(\sqrt{n^3/\log(n)}\cdot \varepsilon\big)} \big)$ (see \Cref{thm:one-cycle} and \Cref{thm:union-cycles}, respectively) are closely related: tight approximations require a large portion of the convergence time, if $\tfrac{\wmax}{\varepsilon}$ dominates~$n$.

\section{Proof of \Cref{thm:one-cycle}}
\label{sec:one-cycle}

We start by motivating some of the key ideas. We first investigate the behavior of BP on the cycle $C_{2n}$ for carefully selected weights. Subsequently we embed $C_{2n}$ into $K_{n,n}$ and complete the argument for \Cref{thm:one-cycle}. 
The cycle $C_{2n}$ on $2n$ nodes (see \Cref{fig:C6} for $n=3$) has two perfect matchings, one of which is optimal, provided that the MWM is unique. Since these two matchings are edge-disjoint, the edges of $C_{2n}$ may be partitioned into \emph{optimal} and \emph{suboptimal} edges.  Now assume that there is a heavy suboptimal edge (see the thick edge $\{\alpha_1, \beta_3\}$ in \Cref{fig:C6}) which is at least twice as heavy as any other edge. It turns out that this heavy edge acts as an \emph{attractor} of suboptimal beliefs. In particular, we show in the \emph{Nibbling Lemma} (\Cref{lemma:nibbling}) that many iterations are required to rule out the heavy edge.

	\begin{figure}[t]
		\centering
		\begin{subfigure}{0.32\linewidth}
		\begin{tikzpicture}[xscale=1, yscale=0.5]	
		\draw node[tinycirc, label=left:\footnotesize $\alpha_1$] (a1) at (0, -3) {};
		\draw node[tinycirc, label=right:\footnotesize $\beta_1$] (b1) at (1, -3) {};	
		\draw node[tinycirc, label=left:\footnotesize $\alpha_2$] (a2) at (0, -4) {};
		\draw node[tinycirc, label=right:\footnotesize $\beta_2$] (b2) at (1, -4) {};	
		\draw node[tinycirc, label=left:\footnotesize $\alpha_3$] (a3) at (0,-5) {};
		\draw node[tinycirc, label=right:\footnotesize $\beta_3$] (b3) at (1,-5) {};

		\draw[bluematch] (a1) -- (b1);	
		\draw[bluematch] (a2) -- (b2);			
		\draw[bluematch] (a3) -- (b3);	
		
		\draw[redheavymatch] (a1) -- (b3);	
		\draw[redmatch] (a2) -- (b1);			
		\draw[redmatch] (a3) -- (b2);	
		\end{tikzpicture}
		\end{subfigure} 
		\begin{subfigure}{0.32\linewidth}
		\begin{tikzpicture}[xscale=1, yscale=0.5]		
		\draw node[tinycirc, label=left:\footnotesize $\alpha_1$] (a1) at (0, -3) {};
		\draw node[tinycirc, label=right:\footnotesize $\beta_1$] (b1) at (1, -3) {};	
		\draw node[tinycirc, label=left:\footnotesize $\alpha_2$] (a2) at (0, -4) {};
		\draw node[tinycirc, label=right:\footnotesize $\beta_2$] (b2) at (1, -4) {};	
		\draw node[tinycirc, label=left:\footnotesize $\alpha_3$] (a3) at (0,-5) {};
		\draw node[tinycirc, label=right:\footnotesize $\beta_3$] (b3) at (1,-5) {};

		\draw[bluematch] (a1) -- (b1);	
		\draw[bluematch] (a2) -- (b2);			
		\draw[bluematch] (a3) -- (b3);		
		\end{tikzpicture}
		\end{subfigure} 
		\begin{subfigure}{0.32\linewidth}
		\begin{tikzpicture}[xscale=1, yscale=0.5]
		\draw node[tinycirc, label=left:\footnotesize $\alpha_1$] (a1) at (0, -3) {};
		\draw node[tinycirc, label=right:\footnotesize $\beta_1$] (b1) at (1, -3) {};	
		\draw node[tinycirc, label=left:\footnotesize $\alpha_2$] (a2) at (0, -4) {};
		\draw node[tinycirc, label=right:\footnotesize $\beta_2$] (b2) at (1, -4) {};	
		\draw node[tinycirc, label=left:\footnotesize $\alpha_3$] (a3) at (0,-5) {};
		\draw node[tinycirc, label=right:\footnotesize $\beta_3$] (b3) at (1,-5) {};
	
		\draw[redheavymatch] (a1) -- (b3);	
		\draw[redmatch] (a2) -- (b1);			
		\draw[redmatch] (a3) -- (b2);	
		\end{tikzpicture}
		\end{subfigure} 
		\caption{left: the cycle $C_{2n}$ (for $n=3$); middle and right: the optimal and suboptimal matching drawn with double and solid edges, respectively.}
		\label{fig:C6}
	\end{figure}
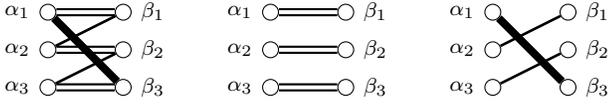		
	
	Now define the edge weights of the cycle graph $C_{2n}$ as follows. Let $n \geq 3$, $[n] \defeq \{1,\dots, n\}$, and $0<\varepsilon < \frac{\wmax}{4(n–2)}$.
	 Denote the layers of the bipartite cycle $C_{2n} \defeq (A_n, B_n, E_n)$  by $A_n \defeq \{\alpha_1, \dots, \alpha_n\}$ and $B_n \defeq \{ \beta_1, \dots, \beta_n \}$. Its edge set is $E_n \defeq  \Eopt \cup \Esub$, where $\Eopt \defeq \big\{ \{\alpha_i, \beta_i\} \mid i \in [n]\big\}$ and $\Esub \defeq \big\{ \{\alpha_{i+1}, \beta_i\} \mid i \in [n–1] \big\} \cup \big\{ \{\alpha_1, \beta_n\} \big\}$ are the sets of \emph{optimal} and \emph{suboptimal} edges, respectively.

	From now on, whenever we refer to $C_{2n}$, its edges are weighted as follows, where $\wmax > 0$ is the largest weight:
	\begin{equation}\label{eq:weights}
	w_e {\defeq} \begin{cases}
	\frac{\wmax}{2} & \text{if } e \in \Eopt, \\
	\tfrac{\wmax}{2} – \tfrac{\wmax}{2(n–1)} – \tfrac{\varepsilon}{n–1}, & \text{if } e = \{\alpha_{i+1}, \beta_i\} {\in} \Esub, \\ 
	\wmax & \text{if } e = \{\alpha_1, \beta_n\}  {\in} \Esub.
	\end{cases}
	\end{equation}
	Note that $\varepsilon \,{<}\, \frac{\wmax}{4(n–2)}$ implies $\tfrac{\wmax}{2} {-} \tfrac{\wmax}{2(n–1)} – \tfrac{\varepsilon}{n–1} \geq 0$ for all $n \,{\geq}\, 3$.
	Let $W(M) \,{\defeq}\, \sum_{e \in M} w_e$ denote the weight of a matching $M$. 
	A simple addition of the edge weights shows that $W(\Eopt) = n {\cdot} \frac{\wmax}{2}$ and $W(\Esub) = (n–1){\cdot}(\tfrac{\wmax}{2} – \tfrac{\wmax}{2(n–1)} – \tfrac{\varepsilon}{n–1}) + \wmax =  W(\Eopt) – \varepsilon$, \ie the maximum weight matching is indeed the set $\Eopt$ of optimal edges, the set $\Esub$ of suboptimal edges is the second best matching, and $\varepsilon$ is the uniqueness gap.
	
	For the remaining analysis of BP on $C_{2n}$, we need some of the concepts and arguments from the proof of \Cref{thm:bayatishahsharma} in~\cite{DBLP:journals/tit/BayatiSS08}. 
	Given an arbitrary graph $G{=}(V,E)$ -- such as $C_{2n}$ or $K_{n,n}$ -- the \emph{computation tree} (or \emph{unwrapped network}) $\T{v}{t}$ of~$v$ at iteration $t$ is constructed as follows: First, let $v$ be the root of $\T{v}{t}$. Then for any node $u$ of $\T{v}{t}$ at depth $t'<t$, make all neighbors of $u$ in $G$ children of $u$ except for its parent in the tree. Note that the depth of $\T{v}{t}$ is exactly $t$. This might differ from other literature where the iteration counter of the BP algorithm starts with $t=0$.

	Now let $\T{v}{t} = (V', E')$. A \emph{T-matching} $T' \subseteq E'$ is a partial matching in $\T{v}{t}$ where every inner node is an endpoint of an edge in $T'$. 
	One can show that the belief of $v$ in $G$ at iteration~$t$ is the same as the belief of $v$ in $\T{v}{t}$ (cf. the unwrapped network lemma in \cite{DBLP:journals/neco/Weiss00}) and that the belief of $v$ in $\T{v}{t}$ is the edge incident with~$v$ in a maximum weight T-matching (cf. Lemma 1 in \cite{DBLP:journals/tit/BayatiSS08}). Thus, for the analysis of BP in $G$, it suffices to only consider maximum weight T-matchings. 	
	
	For $G=C_{2n}$ the situation is simple since every computation tree is a path. 
	Consider the computation tree $\T{v}{t}$ for some iteration $t = kn + \ell$ where $k \geq 0$ and $0 \leq \ell \leq n–1$. Beginning with a leaf, $\T{v}{t}$ is partitioned into $k$ copies of $C_{2n}$ and an incomplete copy, called a \emph{tail}, of $2\ell$ edges (see \Cref{fig:computationtrees}).

	\begin{figure}[t]
		\centering 
		\begin{subfigure}{0.37\linewidth}
		\centering 
		\begin{tikzpicture}
			[xscale=0.5, yscale=0.5]

			\draw node[tinycirc, label=left:\footnotesize $\alpha_1$] (a10) at (0, 0) {};
			\draw node[tinycirc, label=left:\footnotesize $\beta_1$] (b11l) at (-1, -1) {};	
			\draw node[tinycirc, label=right:\footnotesize $\beta_3$] (b21r) at (1, -1) {};	
			\draw node[tinycirc, label=left:\footnotesize $\alpha_2$] (a22l) at (-1, -2) {};	
			\draw node[tinycirc, label=right:\footnotesize $\alpha_3$] (a22r) at (1, -2) {};	
			\draw node[tinycirc, label=left:\footnotesize $\beta_2$] (b23l) at (-1, -3) {};	
			\draw node[tinycirc, label=right:\footnotesize $\beta_2$] (b13r) at (1, -3) {};	
			\draw node[tinycirc, label=left:\footnotesize $\alpha_3$] (a34l) at (-1, -4) {};	
			\draw node[tinycirc, label=right:\footnotesize $\alpha_2$] (a24r) at (1, -4) {};			
			
			\draw [decorate,decoration={brace,amplitude=5pt},xshift=4pt,yshift=0pt]
			(0.5,-4) -- (0.5,-2) node [black,midway,xshift=-0.35cm, rotate=90] 
			{\footnotesize tail};	

			\draw[bluematch] (a10) -- (b11l);	
			\draw[bluematch] (a22l) -- (b23l);	
			\draw[bluematch] (b21r) -- (a22r);	
			\draw[redmatch] (b23l) -- (a34l);
			
			\draw[redheavymatch] (a10) -- (b21r);	
			\draw[redmatch] (b11l) -- (a22l);	
			\draw[redmatch] (a22r) -- (b13r);	
			\draw[bluematch] (b13r) -- (a24r);					
		\end{tikzpicture}
		\end{subfigure}
		\begin{subfigure}{0.37\linewidth}
			\centering 
			\begin{tikzpicture}
			[xscale=0.5, yscale=0.5]

			\draw node[tinycirc, label=left:\footnotesize $\alpha_2$] (a20) at (0, 0) {};
			\draw node[tinycirc, label=left:\footnotesize $\beta_1$] (b11l) at (-1, -1) {};	
			\draw node[tinycirc, label=right:\footnotesize $\beta_2$] (b21r) at (1, -1) {};	
			\draw node[tinycirc, label=left:\footnotesize $\alpha_1$] (a12l) at (-1, -2) {};	
			\draw node[tinycirc, label=right:\footnotesize $\alpha_3$] (a32r) at (1, -2) {};	
			\draw node[tinycirc, label=left:\footnotesize $\beta_3$] (b33l) at (-1, -3) {};	
			\draw node[tinycirc, label=right:\footnotesize $\beta_3$] (b33r) at (1, -3) {};			
			\draw node[tinycirc, label=left:\footnotesize $\alpha_3$] (a34l) at (-1, -4) {};	
			\draw node[tinycirc, label=right:\footnotesize $\alpha_1$] (a14r) at (1, -4) {};	
		
			\draw [decorate,decoration={brace,amplitude=5pt},xshift=4pt,yshift=0pt]
			(0.5,-4) -- (0.5,-2) node [black,midway,xshift=-0.35cm, rotate=90] 
			{\footnotesize tail};
		
			\draw[redmatch] (a20) -- (b11l);	
			\draw[bluematch] (a20) -- (b21r);	
			\draw[bluematch] (a12l) -- (b11l);	
			\draw[redmatch] (b21r) -- (a22r);	
			\draw[bluematch] (a32r) -- (b33r);	
			\draw[redheavymatch] (a12l) -- (b33l);	
			\draw[bluematch] (b33l) -- (a34l);	
			\draw[redheavymatch] (b33r) -- (a14r);	
			\end{tikzpicture}
		\end{subfigure}
		\begin{subfigure}{0.2\linewidth}
		\centering 
		\begin{tikzpicture}
		[xscale=0.5, yscale=0.5]
		
			\draw node[tinycirc] (aiopt) at (0, -1) {};
			\draw node[tinycirc] (biopt) at (2, -1) {};
			\draw[bluematch] (aiopt) -- node[above] {\footnotesize $w_{\text{opt}}$} (biopt);
			
			\draw node[tinycirc] (aisub) at (0, -2.5) {};
			\draw node[tinycirc] (bisub) at (2, -2.5) {};
			\draw[redmatch] (aisub) -- node[above] {\footnotesize $w_{\text{sub}}$} (bisub);			
			
			\draw node[tinycirc] (aiheavy) at (0, -4) {};
			\draw node[tinycirc] (biheavy) at (2, -4) {};
			\draw[redheavymatch] (aiheavy) -- node[above] {\footnotesize $\wmax$} (biheavy);			
		\end{tikzpicture}
	\end{subfigure}	
		\caption{computation trees $\T{\alpha_1}{4}$ and $\T{\alpha_2}{4}$ with edge weights $w_{\text{opt}} = \tfrac{\wmax}{2}$,~ $w_{\text{sub}} = \tfrac{\wmax}{2} {-} \tfrac{\wmax}{2(n–1)} {-} \tfrac{\varepsilon}{n–1}$, and $\wmax$.}
		\label{fig:computationtrees}
	\end{figure}
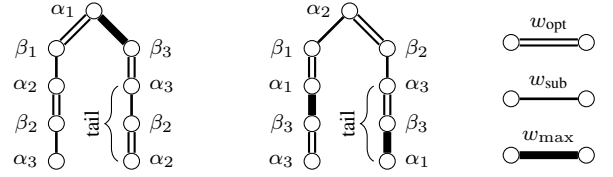
	
	\begin{example}
		Let $n=3$. Consider the cycle $C_6$ and its computation trees $\T{\alpha_1}{4}$ and $\T{\alpha_2}{4}$ as depicted in \Cref{fig:computationtrees}. 
		The maximum weight T-matching in $\T{\alpha_1}{4}$ has the weight $4\cdot \tfrac{\wmax}{2}$ compared to the suboptimal weight of $\wmax + 3\cdot(\tfrac{\wmax}{4} – \tfrac{\varepsilon}{2})$. As a consequence the root $\alpha_1$ of $\T{\alpha_1}{4}$ correctly believes that $\{\alpha_1, \beta_1\}$ belongs to the MWM in $C_6$. On the other hand, 
		the root $\alpha_2$ of $\T{\alpha_2}{4}$ falsely believes that $\{\alpha_2, \beta_1 \}$ is an edge of the MWM in $C_6$
		since the heavy edge in the tail outweighs the $\varepsilon$-advantage of the optimal edges in the copy of $C_6$.
	\end{example}
	\goodbreak 
	
	The following lemma generalizes this observation. 
	
	\begin{lemma}[Nibbling Lemma]\label{lemma:nibbling}
	For every iteration $t = kn + \ell$ with $k \geq 0$ and $1 \leq \ell \leq n–1$, there is a node $v$ such that the computation tree $\T{v}{t}$ consists of $k$ copies of the cycle $C_{2n}$ and a tail of length $2\ell$, where the tail contains the heavy edge $\{\alpha_1, \beta_n\}$. 
	Let $W_\textrm{\normalfont opt}(\T{v}{t})$ and $W_\textrm{\normalfont sub}(\T{v}{t})$ be the weights of the optimal and suboptimal edges in $\T{v}{t}$, respectively. Then
	\begin{equation}\label{eq:nibbling}
			W_\textrm{\normalfont sub}(\T{v}{t}) – W_\textrm{\normalfont opt}(\T{v}{t})  = –k\varepsilon + \Delta_\ell,
	\end{equation}
	where $\tfrac{\wmax}{2} = \Delta_1 > \cdots > \Delta_{n–1} > \tfrac{\wmax}{4(n-1)}$.
	\end{lemma}

	We interpret the Nibbling Lemma as follows. Whenever the tail is nonempty, \ie $\ell \neq 0$, and the tail contains the heavy edge, the suboptimal edges have an advantage of $\Delta_\ell > 0$ in the tail. On the other hand, the higher the number $t$ of iterations, the larger the number $k$ of copies of $C_{2n}$ in $\T{v}{t}$. Since the weight difference is $–k\varepsilon + \Delta_{\ell}$, each of the $k$ copies ``nibbles off'' an~$\varepsilon$ from $\Delta_\ell$. Hence, if $k$ is large enough, $k\varepsilon > \Delta_\ell$ follows for all $\ell \in [n–1]$, and therefore BP converges. 
	
	For the proof of \Cref{thm:one-cycle}, it suffices to consider the case when the tail consists only of the heavy edge and an optimal edge ($\ell = 1$). However, a general version of the Nibbling Lemma is required in the proof of \Cref{thm:union-cycles}. %
	
	\begin{proof}[Proof of \Cref{lemma:nibbling}]
		Let $v$ be some node such that $\T{v}{t}$ contains the heavy edge in its tail. Since $t=kn + \ell$, the computation tree consists of $k$ copies of $C_{2n}$ and a tail of length $2\ell$. The optimal matching on $C_{2n}$ has an advantage of $\varepsilon$ over the suboptimal matching for each copy. However, restricted to the tail, the suboptimal matching wins by 
		\begin{align}
			\Delta_\ell &\stackrel{\eqref{eq:weights}}{\defeq} \left((\ell – 1) {\cdot} \Big(\tfrac{\wmax}{2} – \tfrac{\wmax}{2(n–1)} – \tfrac{\varepsilon}{n–1}\Big) + \wmax\right) – \ell {\cdot} \tfrac{\wmax}{2} \\
			&= \wmax \cdot \tfrac{n–\ell}{2(n–1)} – \varepsilon \cdot \tfrac{\ell–1}{n–1}.
		\end{align}
		Now observe that $\Delta_1 > \dots > \Delta_{n–1}$ is a strictly decreasing sequence which is bounded by $\Delta_1 = \frac{\wmax}{2}$ from above and, due to $\varepsilon < \tfrac{\wmax}{4(n–2)}$ and $n \geq 3$, by
		\begin{align}
			\Delta_{n–1} & =  \tfrac{\wmax}{2(n–1)} – \varepsilon \cdot \tfrac{n–2}{n–1} 
			> \tfrac{\wmax}{2(n–1)} – \tfrac{\wmax}{4(n–1)} = \tfrac{\wmax}{4(n-1)}
		\end{align}
		from below.
 \hfill ${\qedhere_{\text{\Cref{lemma:nibbling}}}}$
	\end{proof}
	\goodbreak
	
	We now show that the weights as defined in~\eqref{eq:weights} force the upper bound in \Cref{thm:bayatishahsharma} to be tight. 
	\begin{proof}[Proof of \Cref{thm:one-cycle}]
	We start our analysis with the graph $C_{2n}$ and explain how 
	the lower bound of $\tfrac{n{\cdot}\wmax}{2\varepsilon}$ for the number of iterations follows from  \Cref{lemma:nibbling}. Consider the largest integer~$k$ such that 
	\begin{align}
		kn+1 < \tfrac{n{\cdot}\wmax}{2\varepsilon}
	\end{align}
	holds. Observe that for $t = kn+1$ there is a node $v$ such that the computation tree $\T{v}{t}$ contains $k$ copies of the cycle $C_{2n}$ and a tail with one copy of $\{\alpha_1, \beta_n\}$ and $\{ \alpha_1, \beta_1 \}$ each. Then~$v$ has a suboptimal belief since $k < \tfrac{\wmax}{2\varepsilon}$ and \eqref{eq:nibbling} from the Nibbling Lemma imply
	\begin{align}
		W_\textrm{\normalfont sub}(\T{v}{t}) – W_\textrm{\normalfont opt}(\T{v}{t})  = –k\varepsilon + \Delta_1
			> –\tfrac{\wmax}{2\varepsilon} {\cdot} \varepsilon + \tfrac{\wmax}{2} = 0.
	\end{align}
	On the other hand, for $k > \tfrac{\wmax}{2\varepsilon}$, BP converges since 
	\begin{align}
	W_\textrm{\normalfont sub}(\T{v}{t}) – W_\textrm{\normalfont opt}(\T{v}{t})  = –k\varepsilon + \Delta_1
	< –\tfrac{\wmax}{2\varepsilon} \cdot \varepsilon + \tfrac{\wmax}{2} = 0.
	\end{align}
	Hence, \Cref{thm:one-cycle} holds for the graph $C_{2n}$.
		
	In order to prove the original version of the theorem, we embed $C_{2n}$ into the complete bipartite graph $K_{n,n}$, where every cycle edge is weighted as in \eqref{eq:weights} and (in a slight abuse of the notation for $\wmax$) every noncycle edge $e$ receives the weight $w_e = –2\wmax$. We call any such edge a \emph{light edge}. 
	
	\begin{proposition}\label{prop:embed-Kpp}
		In every iteration of BP, every node $v$ in $C_{2n}$ has exactly the same belief as $v$ in $K_{n,n}$.
	\end{proposition}
	\begin{figure}[t]
			\centering 
			\begin{tikzpicture}
			[xscale=0.35, yscale=0.7]
			
			\draw[line width=12pt, gray!50, rounded corners=2pt] (-3,-3)  -- (-2,-2) -- (0,-1) -- (0,0) -- (8,-1) -- (10,-2) -- (11,-3);
			
			\draw node[draw=gray!50, circle, line width=0pt, fill=gray!50, inner sep=0pt, minimum width=12pt] (0gray) at (0,0) {~};
			
			\draw node[draw=gray!50, circle, line width=0pt, fill=gray!50, inner sep=0pt, minimum width=12pt] (0gray) at (-3,-3) {~};
			
				\draw node[draw=gray!50, circle, line width=0pt, fill=gray!50, inner sep=0pt, minimum width=12pt] (0gray) at (11,-3) {~};
			
			\draw node[tinycirc, label=above:\footnotesize $\alpha_1$] (0) at (0, 0) {};
			\draw node[tinycirc, label=left:\footnotesize $\beta_1$] (1) at (-8, -1) {};	
			\draw node[tinycirc, label=left:\footnotesize $\beta_2$] (2) at (0, -1) {};	
			\draw node[tinycirc, label=right:\footnotesize $\beta_3$] (3) at (8, -1) {};	
			\draw node[tinycirc, label=left:\footnotesize $\alpha_2$] (4) at (-10, -2){};				
			\draw node[tinycirc, label=left:\footnotesize $\alpha_3$] (5) at (-6, -2){};				
			\draw node[tinycirc, label=left:\footnotesize $\alpha_2$] (6) at (-2, -2){};				
			\draw node[tinycirc, label=left:\footnotesize $\alpha_3$] (7) at (2, -2){};				
			\draw node[tinycirc, label=left:\footnotesize $\alpha_2$] (8) at (6, -2){};				
			\draw node[tinycirc, label=left:\footnotesize $\alpha_3$] (9) at (10, -2){};				
			\draw node[tinycirc, label=below:\footnotesize $\beta_2$] (10) at (-11, -3){};				
			\draw node[tinycirc, label=below:\footnotesize $\beta_3$] (11) at (-9, -3){};				
			\draw node[tinycirc, label=below:\footnotesize $\beta_2$] (12) at (-7, -3){};				
			\draw node[tinycirc, label=below:\footnotesize $\beta_3$] (13) at (-5, -3){};				
			\draw node[tinycirc, label=below:\footnotesize $\beta_1$] (14) at (-3, -3){};				
			\draw node[tinycirc, label=below:\footnotesize $\beta_3$] (15) at (-1, -3){};				
			\draw node[tinycirc, label=below:\footnotesize $\beta_1$] (16) at (1, -3){};				
			\draw node[tinycirc, label=below:\footnotesize $\beta_3$] (17) at (3, -3){};				
			\draw node[tinycirc, label=below:\footnotesize $\beta_1$] (18) at (5, -3){};				
			\draw node[tinycirc, label=below:\footnotesize $\beta_2$] (19) at (7, -3){};				
			\draw node[tinycirc, label=below:\footnotesize $\beta_1$] (20) at (9, -3){};				
			\draw node[tinycirc, label=below:\footnotesize $\beta_2$] (21) at (11, -3){};
			
			\draw[bluematch] (0) -- (1);
			\draw[light] (0) -- (2);	
			\draw[redheavymatch] (0) -- (3);
			
			\draw[redmatch] (1) -- (4);
			\draw[light] (1) -- (5);
			\draw[bluematch] (2) -- (6);
			\draw[redmatch] (2) -- (7);
			\draw[light] (3) -- (8);
			\draw[bluematch] (3) -- (9);
			
			\draw[bluematch] (4) -- (10);
			\draw[light] (4) -- (11);
			\draw[redmatch] (5) -- (12);
			\draw[bluematch] (5) -- (13);
			\draw[redmatch] (6) -- (14);
			\draw[light] (6) -- (15);
			\draw[light] (7) -- (16);
			\draw[bluematch] (7) -- (17);
			\draw[redmatch] (8) -- (18);
			\draw[bluematch] (8) -- (19);
			\draw[light] (9) -- (20);
			\draw[redmatch] (9) -- (21);

			\end{tikzpicture}
			\caption{the augmenting path argument from \Cref{prop:embed-Kpp} where light edges are depicted with dotted edges; suppose a T-matching contained the light edge $\{\alpha_1, \beta_2\}$ at the root; then flipping the edges along the path increases the weight of the T-matching.}
			\label{fig:illustration-of-proposition}
	\end{figure}
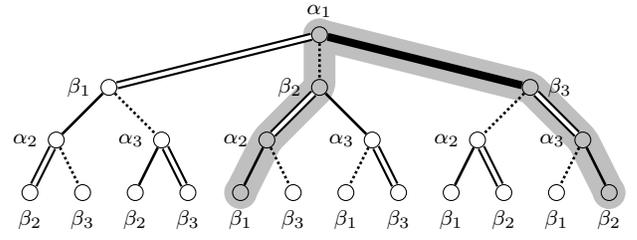
	
	Note that the argument holds for non-negative weights as well when $2\wmax$ is added to the weight of each edge.
	\begin{proof}[Proof of \Cref{prop:embed-Kpp} (sketch)]
		We show that for every computation tree, every maximum weight T-matching does not contain a light edge. Otherwise the weight of this T-matching could be increased using an augmenting path argument (see \Cref{fig:illustration-of-proposition}) where the path contains the light edge, as well as suboptimal and optimal edges alternately. Hence, maximum weight T-matchings only contain cycle edges.
		\hfill $\qedhere_{\text{\Cref{prop:embed-Kpp}}}$
	\end{proof}
	\noindent Now the claim follows for $K_{n,n}$.		
	\hfill $\qedhere_{\text{Theorem 2}}$
	\end{proof}

	\section{Proof of \Cref{thm:union-cycles}}\label{sec:union-cycles}
		
	\Cref{thm:one-cycle}, in conjunction with the upper bound of Bayati, Shah, Sharma (\Cref{thm:bayatishahsharma}), characterizes the worst-case convergence time. However, BP already computes the maximum weight matching in $C_{2n}$ after performing $t=n$ iterations since every computation tree at time $n$ does not have a tail. 
	Instead of using a single cycle, we build a graph using multiple node-disjoint cycles of length $2n_i$ where $n_1 <  \dots < n_c$ are prime numbers with the same order of magnitude. The convergence time for cycle $C_{2n_i}$ coincides with $\tfrac{2n_i\wmax}{\varepsilon}$, but the construction prevents BP from finding a perfect matching in unions of cycles as an intermediate solution. Based on this observation, we will see that even partial intermediate solutions cannot be completed to give matchings with a weight close to the weight of the MWM.
	
		\goodbreak 
	\begin{proof}[Proof of \Cref{thm:union-cycles}]
	We begin our investigation on the cycle construction with the following lemma.
	
	\begin{lemma}[Dusart, \cite{dusart1998autour}]\label{lemma:primecounting}
		For every $n \geq 599$, the number $\pi(n)$ of prime numbers less than or equal $n$ is bounded by
		\begin{align}
			\tfrac{n}{\log(n)} \left(1+\tfrac{1}{\log(n)} \right)  \leq  \pi(n) \leq  \tfrac{n}{\log(n)} \left(1+\tfrac{1.2762}{\log(n)} \right)
		\end{align}

	\end{lemma}

	For the rest of the proof let $n$ be sufficiently large and $c \leq \tfrac{1}{2}\sqrt{{n}/{\log(n)}}$. 
	Let $P_{n,c}$ denote the set of prime numbers in the open interval $(\tfrac{n}{2c}, \tfrac{n}{c})$. We apply \Cref{lemma:primecounting}, obtain $|P_{n,c}| = \pi(\tfrac{n}{c}) – \pi(\tfrac{n}{2c}) >  \tfrac{n}{4c\log(n)}$, and  $|P_{n,c}| \geq c$ follows.
	Now select $c$ prime numbers $n_1 <  {\cdots} < n_c$ from $P_{n,c}$ and let $C_{2n_1}, \dots, C_{2n_c}$ be node-disjoint cycles with weights as described in \eqref{eq:weights}. 

	The next lemma states that BP fails for many cycles within a large number of iterations. 
	\begin{lemma}\label{lemma:technical-union-cycles} For the cycles $C_{n_1}, \dots, C_{n_c}$ the following holds:
If $t \leq \min \{ \tfrac{\wmax}{8c\varepsilon}, \lfloor (\tfrac{n}{2c})^{\nicefrac{c}{2}} \rfloor \}$, there are at least $\tfrac{c}{2}$ cycles such that BP does not find a perfect matching for any of these cycles.
	\end{lemma}
	\begin{proof}[Proof of \Cref{lemma:technical-union-cycles}]
		Let $n_i \in P_{n,c}$. If $t \not\equiv 0 \mod n_i$ and $t \leq \tfrac{n_i {\cdot}\wmax}{2\varepsilon}$, then \Cref{lemma:nibbling} implies that BP does not find a perfect matching for $C_{2n_i}$ at iteration $t$. 
		Now note that for every $t \leq \min \{ \tfrac{\wmax}{8c\varepsilon}, \lfloor (\tfrac{n}{2c})^{\nicefrac{c}{2}} \rfloor \}$, the prime factorization of $t$ contains at most~$\tfrac{c}{2}$ distinct prime numbers from $P_{n,c}$. Hence, there is a set $Q \subseteq P_{n,c}$ of at least~$\tfrac{c}{2}$ prime numbers such that $t \not\equiv 0 \bmod n_j$ for all $n_j \in Q$. Now the claim follows. 
		~\hfill $\qedhere_{\text{\Cref{lemma:technical-union-cycles}}}$
	\end{proof}

	By embedding $\bigcup_{i=1}^{c} C_{2n_i}$ into the $K_{n,n}$ such that the arguments from \Cref{lemma:technical-union-cycles} still hold, we reach another important milestone in our reasoning. 
	W.l.o.g. we assume $n = \sum_{ i=1}^c n_i$; otherwise extend $K_{n',n'}$, where $n'=\sum_{ i=1}^c n_i$, with a matching on $2(n–n')$ new nodes and let each new matching edge $e$ receive the weight $w_e = \tfrac{\wmax}{2}$. Finally, weight every other edge~$e'$ in $K_{n,n}$ with $w_{e'} = –2\wmax$.
	
	\begin{proposition}\label{prop:embed-KNN}
		In every iteration of BP, every node $v$ in $\bigcup_{i=1}^c C_{2n_i}$ has exactly the same belief as $v$ in~$K_{n,n}$.
	\end{proposition}
	\begin{proof}[Proof of \Cref{prop:embed-KNN}]
		The proof is analogous to the proof of \Cref{prop:embed-Kpp}.
		~\hfill $\qedhere_{\text{\Cref{prop:embed-Kpp}}}$
	\end{proof}
	
	Hence, \Cref{lemma:technical-union-cycles} implies that BP fails to find perfect matchings for at least $\tfrac{c}{2}$ node-disjoint cycles in $K_{n,n}$. In order to gain a better understanding of completing partial BP matchings, the next example illustrates the exact behavior of BP for our constructed weights.

	\begin{figure}[t]
		\centering 	
		\begin{subfigure}{0.3\linewidth}
			\centering 
			\begin{tikzpicture}[xscale=1, yscale=0.5]
			\tikzset{
				bluematch/.style   = {double, double distance=1.5pt, thick, >=stealth },				
				redheavymatch/.style   = {line width=2.5pt },
				redmatch/.style   = {line width=1pt, >=stealth },
				tinycirc/.style = {draw, circle, inner sep=0pt, minimum width=6pt}					
			}

			\draw node[tinycirc, label=left:\footnotesize $\alpha_1$] (a1) at (0,-3) {};
			\draw node[tinycirc, label=right:\footnotesize $\beta_1$, fill=black] (b1) at (1,-3) {};	
			\draw node[tinycirc, label=left:\footnotesize $\alpha_2$] (a2) at (0,-4) {};
			\draw node[tinycirc, label=right:\footnotesize$\beta_2$] (b2) at (1,-4) {};	
			\draw node[tinycirc, label=left:\footnotesize $\alpha_3$] (a3) at (0,-5) {};
			\draw node[tinycirc, label=right:\footnotesize$\beta_3$] (b3) at (1,-5) {};
			\draw node[tinycirc, label=left:\footnotesize $\alpha_4$] (a4) at (0,-6) {};
			\draw node[tinycirc, label=right:\footnotesize$\beta_4$] (b4) at (1,-6) {};
			\draw node[tinycirc, label=left:\footnotesize $\alpha_5$, fill=black] (a5) at (0,-7) {};
			\draw node[tinycirc, label=right:\footnotesize$\beta_5$] (b5) at (1,-7) {};

			\draw[bluematch] (a2) -- (b2);
			\draw[bluematch] (a3) -- (b3);	
			\draw[bluematch] (a4) -- (b4);	
			\draw[bluematch,->] (b1) -- (a1);
			\draw[bluematch,->] (a5) -- (b5);

			\draw[redheavymatch] (a1) -- (b5);

			\end{tikzpicture}
			\caption{$t=1$}
		\end{subfigure}
		\begin{subfigure}{0.3\linewidth}
			\centering 
			\begin{tikzpicture}[xscale=1, yscale=0.5]
			\tikzset{
				bluematch/.style   = {double, double distance=1.5pt, thick, >=stealth },				
				redheavymatch/.style   = {line width=2.5pt },
				redmatch/.style   = {line width=1pt, >=stealth },
				tinycirc/.style = {draw, circle, inner sep=0pt, minimum width=6pt}					
			}

			\draw node[tinycirc, label=left:\footnotesize $\alpha_1$] (a1) at (0,-3) {};
			\draw node[tinycirc, label=right:\footnotesize $\beta_1$, fill=black] (b1) at (1,-3) {};	
			\draw node[tinycirc, label=left:\footnotesize $\alpha_2$] (a2) at (0,-4) {};
			\draw node[tinycirc, label=right:\footnotesize $\beta_2$] (b2) at (1,-4) {};	
			\draw node[tinycirc, label=left:\footnotesize $\alpha_3$] (a3) at (0,-5) {};
			\draw node[tinycirc, label=right:\footnotesize $\beta_3$] (b3) at (1,-5) {};
			\draw node[tinycirc, label=left:\footnotesize $\alpha_4$] (a4) at (0,-6) {};
			\draw node[tinycirc, label=right:\footnotesize $\beta_4$] (b4) at (1,-6) {};
			\draw node[tinycirc, label=left:\footnotesize $\alpha_5$, fill=black] (a5) at (0,-7) {};
			\draw node[tinycirc, label=right:\footnotesize $\beta_5$] (b5) at (1,-7) {};

			\draw[bluematch] (a2) -- (b2);
			\draw[bluematch] (a3) -- (b3);	
			\draw[bluematch] (a4) -- (b4);	

			\draw[redheavymatch] (a1) -- (b5);
			\draw[redmatch,->] (b1) -- (a2);
			\draw[redmatch,->] (a5) -- (b4);

			\end{tikzpicture}
			\caption{$t=2$}
		\end{subfigure}
		\begin{subfigure}{0.3\linewidth}
			\centering 
			\begin{tikzpicture}[xscale=1, yscale=0.5]
			\tikzset{
				bluematch/.style   = {double, double distance=1.5pt, thick, >=stealth },				
				redheavymatch/.style   = {line width=2.5pt },
				redmatch/.style   = {line width=1pt, >=stealth },
				tinycirc/.style = {draw, circle, inner sep=0pt, minimum width=6pt}					
			}

			\draw node[tinycirc, label=left:\footnotesize $\alpha_1$] (a1) at (0,-3) {};
			\draw node[tinycirc, label=right:\footnotesize $\beta_1$] (b1) at (1,-3) {};	
			\draw node[tinycirc, label=left:\footnotesize $\alpha_2$] (a2) at (0,-4) {};
			\draw node[tinycirc, label=right:\footnotesize $\beta_2$, fill=black] (b2) at (1,-4) {};	
			\draw node[tinycirc, label=left:\footnotesize $\alpha_3$] (a3) at (0,-5) {};
			\draw node[tinycirc, label=right:\footnotesize $\beta_3$] (b3) at (1,-5) {};
			\draw node[tinycirc, label=left:\footnotesize $\alpha_4$, fill=black] (a4) at (0,-6) {};
			\draw node[tinycirc, label=right:\footnotesize $\beta_4$] (b4) at (1,-6) {};
			\draw node[tinycirc, label=left:\footnotesize $\alpha_5$] (a5) at (0,-7) {};
			\draw node[tinycirc, label=right:\footnotesize $\beta_5$] (b5) at (1,-7) {};

			\draw[bluematch, ->] (b2) -- (a2);
			\draw[bluematch] (a3) -- (b3);	
			\draw[bluematch,->] (a4) -- (b4);	

			\draw[redheavymatch] (a1) -- (b5);
			\draw[redmatch] (b1) -- (a2);
			\draw[redmatch] (a5) -- (b4);

			\end{tikzpicture}
			\caption{$t=3$}
		\end{subfigure}
		\vspace{0.5cm}
		
		\begin{subfigure}{0.3\linewidth}
			\centering 
			\begin{tikzpicture}[xscale=1, yscale=0.5]
			\tikzset{
				bluematch/.style   = {double, double distance=1.5pt, thick, >=stealth },				
				redheavymatch/.style   = {line width=2.5pt },
				redmatch/.style   = {line width=1pt, >=stealth },
				tinycirc/.style = {draw, circle, inner sep=0pt, minimum width=6pt}					
			}

			\draw node[tinycirc, label=left:\footnotesize $\alpha_1$] (a1) at (0,-3) {};
			\draw node[tinycirc, label=right:\footnotesize $\beta_1$] (b1) at (1,-3) {};	
			\draw node[tinycirc, label=left:\footnotesize $\alpha_2$] (a2) at (0,-4) {};
			\draw node[tinycirc, label=right:\footnotesize $\beta_2$, fill=black] (b2) at (1,-4) {};	
			\draw node[tinycirc, label=left:\footnotesize $\alpha_3$] (a3) at (0,-5) {};
			\draw node[tinycirc, label=right:\footnotesize $\beta_3$] (b3) at (1,-5) {};
			\draw node[tinycirc, label=left:\footnotesize $\alpha_4$, fill=black] (a4) at (0,-6) {};
			\draw node[tinycirc, label=right:\footnotesize $\beta_4$] (b4) at (1,-6) {};
			\draw node[tinycirc, label=left:\footnotesize $\alpha_5$] (a5) at (0,-7) {};
			\draw node[tinycirc, label=right:\footnotesize $\beta_5$] (b5) at (1,-7) {};

			\draw[bluematch] (a3) -- (b3);

			\draw[redheavymatch] (a1) -- (b5);
			\draw[redmatch] (b1) -- (a2);
			\draw[redmatch, ->] (b2) -- (a3);
			\draw[redmatch] (a5) -- (b4);
			\draw[redmatch, ->] (a4) -- (b3);

			\end{tikzpicture}
			\caption{$t=4$}
		\end{subfigure}	
		\begin{subfigure}{0.3\linewidth}
			\centering 
			\begin{tikzpicture}[xscale=1, yscale=0.5]
			\tikzset{
				bluematch/.style   = {double, double distance=1.5pt, thick, >=stealth },				
				redheavymatch/.style   = {line width=2.5pt },
				redmatch/.style   = {line width=1pt, >=stealth },
				tinycirc/.style = {draw, circle, inner sep=0pt, minimum width=6pt}					
			}

			\draw node[tinycirc, label=left:\footnotesize $\alpha_1$] (a1) at (0,-3) {};
			\draw node[tinycirc, label=right:\footnotesize $\beta_1$] (b1) at (1,-3) {};	
			\draw node[tinycirc, label=left:\footnotesize $\alpha_2$] (a2) at (0,-4) {};
			\draw node[tinycirc, label=right:\footnotesize $\beta_2$] (b2) at (1,-4) {};	
			\draw node[tinycirc, label=left:\footnotesize $\alpha_3$] (a3) at (0,-5) {};
			\draw node[tinycirc, label=right:\footnotesize $\beta_3$] (b3) at (1,-5) {};
			\draw node[tinycirc, label=left:\footnotesize $\alpha_4$] (a4) at (0,-6) {};
			\draw node[tinycirc, label=right:\footnotesize $\beta_4$] (b4) at (1,-6) {};
			\draw node[tinycirc, label=left:\footnotesize $\alpha_5$] (a5) at (0,-7) {};
			\draw node[tinycirc, label=right:\footnotesize $\beta_5$] (b5) at (1,-7) {};

			\draw[bluematch] (a1) -- (b1);	
			\draw[bluematch] (a2) -- (b2);			
			\draw[bluematch] (a3) -- (b3);	
			\draw[bluematch] (a4) -- (b4);	
			\draw[bluematch] (a5) -- (b5);	
			
			\end{tikzpicture}
			\caption{$t=5$}
		\end{subfigure}	
		\begin{subfigure}{0.3\linewidth}
			\centering 
			\begin{tikzpicture}[xscale=1, yscale=0.5]
			\tikzset{
				bluematch/.style   = {double, double distance=1.5pt, thick, >=stealth },				
				redheavymatch/.style   = {line width=2.5pt },
				redmatch/.style   = {line width=1pt, >=stealth },
				tinycirc/.style = {draw, circle, inner sep=0pt, minimum width=6pt}					
			}

			\draw node[tinycirc, label=left:\footnotesize $\alpha_1$] (a1) at (0,-3) {};
			\draw node[tinycirc, label=right:\footnotesize $\beta_1$, fill=black] (b1) at (1,-3) {};	
			\draw node[tinycirc, label=left:\footnotesize $\alpha_2$] (a2) at (0,-4) {};
			\draw node[tinycirc, label=right:\footnotesize$\beta_2$] (b2) at (1,-4) {};	
			\draw node[tinycirc, label=left:\footnotesize $\alpha_3$] (a3) at (0,-5) {};
			\draw node[tinycirc, label=right:\footnotesize$\beta_3$] (b3) at (1,-5) {};
			\draw node[tinycirc, label=left:\footnotesize $\alpha_4$] (a4) at (0,-6) {};
			\draw node[tinycirc, label=right:\footnotesize$\beta_4$] (b4) at (1,-6) {};
			\draw node[tinycirc, label=left:\footnotesize $\alpha_5$, fill=black] (a5) at (0,-7) {};
			\draw node[tinycirc, label=right:\footnotesize$\beta_5$] (b5) at (1,-7) {};

			\draw[bluematch] (a2) -- (b2);
			\draw[bluematch] (a3) -- (b3);	
			\draw[bluematch] (a4) -- (b4);	
			\draw[bluematch, ->] (b1) -- (a1);
			\draw[bluematch,->] (a5) -- (b5);

			\draw[redheavymatch] (a1) -- (b5);
			
			\end{tikzpicture}
			\caption{$t=6$}
		\end{subfigure}	
		
		\caption{beliefs in $C_{10}$ for iteration $t=1,\dots, 6$; see \Cref{ex:pos-harm-c10}.}
		\label{fig:four-inconsistency}
	\end{figure}
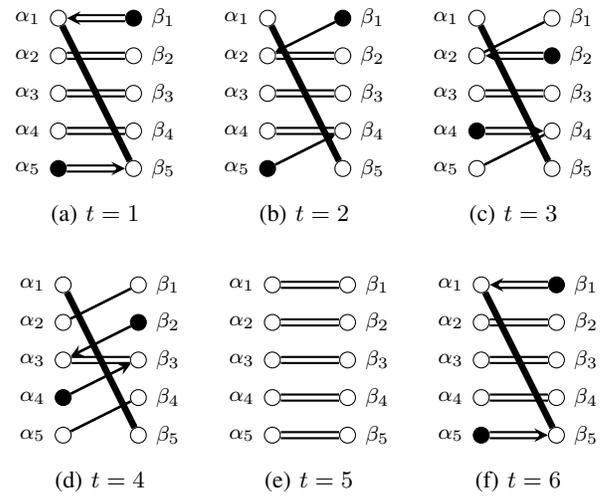

	\begin{example}\label{ex:pos-harm-c10}
		Consider the cycle $C_{10}$. The beliefs at iteration $t=1,\dots,6$ are shown in \Cref{fig:four-inconsistency} where each undirected edge $\{u,v\}$ indicates that both endpoints believe in $\{u,v\}$ belonging to the MWM; and where each directed edge $(u,v)$ indicates that $u$ believes in $\{u,v\}$ belonging to the MWM, but $v$ does not. With increasing $t$, the number of optimal edges in the partial BP matching decreases and the number of suboptimal edges increases. However, in each iteration $t=1,\dots, 4$, there are only two nodes that are not endpoints of a partial matching. Finally, BP finds the optimal matching at iteration $t=5$. For larger $t$, the beliefs repeat periodically until the process converges. 
	\end{example}
	
	\begin{lemma}\label{lemma:union-cycles}
		For every iteration $t \leq \min \{ \tfrac{\wmax}{8c\varepsilon},  \lfloor (\tfrac{n}{2c})^{\nicefrac{c}{2}} \rfloor \}$ and every completion of a partial BP matching, its weight is at most $\big(1 – \Theta(\frac{c}{n})\big)\cdot W_{\text{opt}}$, where $W_{\text{opt}} = n {\cdot} \tfrac{\wmax}{2}$ is the weight of the MWM  for $K_{n,n}$.
	\end{lemma}

	\begin{proof}[Proof of \Cref{lemma:union-cycles}]
		As a consequence of \Cref{lemma:technical-union-cycles} and  the observation we made in \Cref{ex:pos-harm-c10}, there is a set $Q \subseteq P_{n,c}$ of at least $\tfrac{c}{2}$ prime numbers such that for each $n_i \in Q$, the partial BP matching for $C_{2n_i}$ consists of $n_i – 1$ edges. 
		In order to complete the partial BP matching for one of those cycles, we are forced to add a light edge, \ie an edge~$e$ (between two black nodes in \Cref{fig:four-inconsistency}) with weight $w_e = –2\wmax$. In iteration $t \equiv 1 \bmod n_i$, the completion has the highest weight, namely $– 2\wmax + (n_i–2) \cdot \tfrac{\wmax}{2} + \wmax = – 2\wmax + n_i \cdot \tfrac{\wmax}{2}$. Thus the completion for $C_{2n_i}$ has a weight of at most $\big(1 – \Theta(\tfrac{1}{n_i})\big) \cdot W^{(i)}_\text{opt}$, where $W^{(i)}_\text{opt} = n_i \cdot \tfrac{\wmax}{2}$ is the weight of the MWM restricted to $C_{2n_i}$. 
		In total, completing partial BP matchings for $K_{n,n}$ is at most~$1 – \Theta(\tfrac{c}{n})$-approximative. 
			~\hfill $\qedhere_{\text{\Cref{lemma:union-cycles}}}$
	\end{proof}
	
	A worst-case analysis of \Cref{lemma:union-cycles} concludes the proof of \Cref{thm:union-cycles}. 
	In order to push $\lfloor (\tfrac{n}{2c})^c \rfloor$ as close as possible to the convergence bound, we are interested in the largest $c$ such that $\lfloor (\tfrac{n}{2c})^c \rfloor \leq \tfrac{\wmax}{8c\varepsilon} \leq \lfloor (\tfrac{n}{2(c+1)})^{c+1} \rfloor$ holds. Observe that the left-hand and right-hand side of this inequation differ at most by the factor~$n$. Hence, $\lfloor (\tfrac{n}{2c})^c \rfloor \geq \tfrac{\wmax}{8nc\varepsilon}$, \ie we lose the factor $\tfrac{1}{16n^2c}$ of the Bayati-Shah-Sharma convergence bound. 
	Now \Cref{thm:union-cycles} follows by plugging in $c = \tfrac{1}{2} \sqrt{n/\log(n)}$ into \Cref{lemma:union-cycles}.
	~\hfill $\qedhere_{\text{\Cref{thm:union-cycles}}}$
\end{proof}

\section{Approximate Belief Propagation}\label{sec:approximate-bp}
	Finally, we present a linear-time algorithm for the completion of partial matchings which ``respects'' the beliefs of the nodes and only adds edges. However, since the analysis in the proof of \Cref{thm:union-cycles} is not restricted to any specific algorithm, the algorithm described here cannot improve its approximation factor.
	
We call a pair $(\alpha,\beta)$ a \emph{conflict} if exactly one of the two nodes believes that $\{\alpha, \beta\}$ belongs to the MWM. 
For each BP-iteration $t$ for the $K_{n,n} = (A_n, B_n, E_n)$, let $C^{\text{BP}(t)} \defeq (A, B, E_t)$ be the bipartite \emph{conflict graph} with 
\begin{align*}
	A &\defeq \{\alpha \mid \alpha \text{ is not covered by the partial BP matching}\}, \\
	B &\defeq \{\beta \mid \beta \text{ is not covered by the partial BP matching}\}, \\
	E_t &\defeq \big\{  \{ \alpha, \beta \} \mid (\alpha,\beta) \text{ is a conflict}  \big\}.
\end{align*}

\Cref{fig:belief-graph} shows the transformation of beliefs into a conflict graph. 
Note that every connected component of a conflict graph for the assignment problem has at most one cycle. 

\begin{figure}[t]
	\begin{subfigure}{0.48\linewidth}
		\centering 
		\begin{tikzpicture}[xscale=1, yscale=0.5]		
		\draw node[tinycirc, label=left:$\alpha_1$, fill=black] (a1) at (0,  0) {};
		\draw node[tinycirc, label=left:$\alpha_2$, fill=black] (a2) at (0, -1) {};
		\draw node[tinycirc, label=left:$\alpha_3$, fill=black] (a3) at (0, -2) {};
		\draw node[tinycirc, label=left:$\alpha_4$, fill=black] (a4) at (0, -3) {};
		\draw node[tinycirc, label=left:$\alpha_5$] (a5) at (0, -4) {};
		\draw node[tinycirc, label=left:$\alpha_6$, fill=black] (a6) at (0, -5) {};
		
		\draw node[tinycirc, label=right:$\beta_1$, fill=black] (b1) at (1,  0) {};
		\draw node[tinycirc, label=right:$\beta_2$, fill=black] (b2) at (1, -1) {};
		\draw node[tinycirc, label=right:$\beta_3$, fill=black] (b3) at (1, -2) {};
		\draw node[tinycirc, label=right:$\beta_4$, fill=black] (b4) at (1, -3) {};
		\draw node[tinycirc, label=right:$\beta_5$] (b5) at (1, -4) {};
		\draw node[tinycirc, label=right:$\beta_6$, fill=black] (b6) at (1, -5) {};
		
		\draw[>=stealth, ->] (a1) -- (b2);
		\draw[>=stealth, ->] (a2) -- (b2);
		\draw[>=stealth, ->] (a3) -- (b3);
		\draw[>=stealth, ->] (a4) -- (b4);
		\draw[>=stealth] (a5) -- (b5);
		\draw[>=stealth, ->] (a6) -- (b4);
		
		\draw[>=stealth, ->] (b1) -- (a2);
		\draw[>=stealth, ->] (b2) -- (a3);
		\draw[>=stealth, ->] (b3) -- (a4);
		\draw[>=stealth, ->] (b4) -- (a3);
		\draw[>=stealth] (b5) -- (a5);
		\draw[>=stealth, ->] (b6) -- (a5);
		\end{tikzpicture}
		\label{fig:belief-graph-a}
	\end{subfigure}
	\begin{subfigure}{0.48\linewidth}
		\centering 
		\begin{tikzpicture}[xscale=1, yscale=0.5]
		
		\draw node[tinycirc, label=left:$\alpha_1$, fill=black] (a1) at (0,  0) {};
		\draw node[tinycirc, label=left:$\alpha_2$, fill=black] (a2) at (0, -1) {};
		\draw node[tinycirc, label=left:$\alpha_3$, fill=black] (a3) at (0, -2) {};
		\draw node[tinycirc, label=left:$\alpha_4$, fill=black] (a4) at (0, -3) {};
		\draw node[tinycirc, label=left:$\alpha_6$, fill=black] (a6) at (0, -5) {};
		
		\draw node[tinycirc, label=right:$\beta_1$, fill=black] (b1) at (1,  0) {};
		\draw node[tinycirc, label=right:$\beta_2$, fill=black] (b2) at (1, -1) {};
		\draw node[tinycirc, label=right:$\beta_3$, fill=black] (b3) at (1, -2) {};
		\draw node[tinycirc, label=right:$\beta_4$, fill=black] (b4) at (1, -3) {};
		\draw node[tinycirc, label=right:$\beta_6$, fill=black] (b6) at (1, -5) {};
		
		\draw[-] (a1) -- (b2);
		\draw[-] (a2) -- (b2);
		\draw[-] (a3) -- (b3);
		\draw[-] (a4) -- (b4);
		
		\draw[-] (a6) -- (b4);
		
		\draw[-] (b1) -- (a2);
		\draw[-] (b2) -- (a3);
		\draw[-] (b3) -- (a4);
		\draw[-] (b4) -- (a3);

		\end{tikzpicture}
		\label{fig:belief-graph-b}
	\end{subfigure}
	\caption{graphical representation of beliefs and their conflict graph; the edge $\{\alpha_5,\beta_5\}$ belongs the partial BP matching and hence, both endpoints do not occur in the conflict graph.}
	\label{fig:belief-graph}
\end{figure}
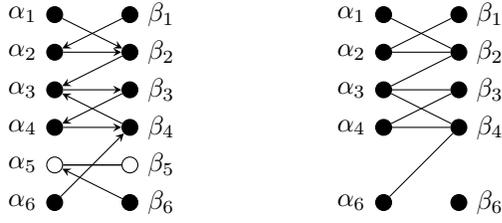

\goodbreak 
\enlargethispage{-0.6cm}

W.l.o.g. let the conflict graph be connected and have a cycle. For every iteration $t$, let $M_t$ be the approximate MWM in $K_{n,n}$ initialized with the partial BP matching. 
For an arbitrary cycle edge $e$ consider the following two a-posteriori cases: 
\begin{enumerate}[topsep=0pt, itemsep=0pt]
	\item[(a)] $e$ belongs to $M_t$; then remove $e$ and its incident edges from the conflict graph;
	\item[(b)] $e$ does not belong to $M_t$; then remove $e$ as well.
\end{enumerate}
In either case the resulting graph is a forest (or even a tree).  
We execute BP for both forests and obtain maximum weight T-matchings  $M_a$ and $M_b$ since BP is correct for trees (see  \cite[Theorem 14.1]{mezard2009information} for a detailed proof). 
If $W(M_a) > W(M_b)$, then set $M_t \defeq M_t \cup M_a \cup \{e\}$, otherwise set $M_t \defeq M_t \cup M_b$. Now remove the edges in $M_t$ and their endpoints from the conflict graph. 

We still have to worry about matching the remaining leafs, denoted by the subsets $A'$ and $B'$.
Observe that $|A'| = |B'|$ and that the set of edges between $A'$ and $B'$ in the conflict graph is empty. 
Compute an arbitrary matching $M'$ between $A'$ and $B'$ with edges from $K_{n,n}$, \eg by using a greedy algorithm, and set $M_t \defeq M_t \cup M'$. 
Finally, output the approximate MWM~$M_t$.

For the weights that we used in the proofs of \Cref{thm:one-cycle} and \Cref{thm:union-cycles}, this algorithm is trivial since the conflict graph consists of isolated nodes only. However, even though approximate BP cannot improve upon the $1–\Theta(\tfrac{c}{n})$ barrier from \Cref{lemma:union-cycles}, we suggest that similar algorithms should also be of interest for the application of BP to other combinatorial optimization problems. 
\goodbreak 

\section{Conclusions}
\label{sec:conclusions}
We established lower bounds on the running time of the BP algorithm for the assignment problem. With respect to convergence, \Cref{thm:one-cycle} states that the upper bound of $2n {\cdot} \wmax/\varepsilon$ on the number of iterations (see \Cref{thm:bayatishahsharma}) is tight up to a factor of four. 
\Cref{thm:union-cycles} considers the behavior of BP when convergence is not required. There are edge weights for complete bipartite graphs such that tight BP-based approximations consume a large portion of the convergence time. The exact number of iterations for a $1-\nicefrac{1}{\sqrt{n/\log(n)}}$-approximate solution belongs to the interval $\Theta(  \sqrt{\log(n)/n^3} {\cdot} \tfrac{\wmax}{\varepsilon}), \dots, 2n{\cdot}\tfrac{\wmax}{\varepsilon}$. We have to leave its exact value open. 
Possibly, a tight analysis of the approximation time requires a construction different from our cycle construction. 

We proposed an approximate BP algorithm which has the advantage of outputting a (suboptimal) solution in every iteration. 
Also, similar lower bounds for other applications of BP to combinatorial optimization problems remain an open research question. An upper bound for the convergence time for the MWM problem for non-bipartite graphs -- under certain restrictions -- is shown in~\cite{DBLP:journals/tit/AhnCGPS18}. Our methods can be utilized to provide a tight runtime analysis in this case too.
Finally we pose the question, under which circumstances can the proposed approximate BP algorithm be used as a tool when BP does not converge or when the underlying decision problem is computationally hard?

\bibliographystyle{IEEEtran}

\end{document}